\documentclass[11pt,onecolumn]{IEEEtran}
\usepackage{amsthm}
\usepackage{times,amssymb,amsmath,amsfonts,float,nicefrac,color,bbm,mathrsfs,float}
\usepackage{enumerate,multirow,caption,tikz,graphicx}
\usepackage{algpseudocode}
\usepackage{stackengine} 
\usepackage[noadjust]{cite}
\usepackage{subcaption}
\usepackage{diagbox}
\usepackage[ruled,vlined,linesnumbered]{algorithm2e}

\usepackage[top=1in,bottom=1in,left=1in,right=1in]{geometry}
\allowdisplaybreaks

\newcommand{\RN}[1]{%
	\textup{\expandafter{\romannumeral#1}}%
}

\usetikzlibrary{shapes.misc, positioning}
\usetikzlibrary{decorations.pathreplacing}

\tikzset{
	block/.style   = {draw, thick, rectangle, minimum width = 1em},
	sblock/.style  = {draw, thick, rectangle, minimum height = 6em, minimum width = 6em}, 
	lgblock/.style = {draw, thick, rectangle, minimum height = 9em, minimum width = 6em}, 	
	vblock/.style  = {draw, thick, rectangle, minimum height = 3.7em, minimum width = 1.8em}, 
}

\tikzset{XOR/.style={draw,circle,append after command={
			[shorten >=\pgflinewidth, shorten <=\pgflinewidth,]
			(\tikzlastnode.north) edge (\tikzlastnode.south)
			(\tikzlastnode.east) edge (\tikzlastnode.west)
		}
	}
}

\newcommand\remove[1]{}

\allowdisplaybreaks

\newtheorem{theorem}{Theorem}

\newtheorem{lemma}{Lemma}

\newtheorem{cnstr}{Construction}

\newcommand{\Fq}{\mathbb{F}_{q}}

\newcommand{\ff}{{\mathbb F}}

\newcommand{\bP}{\mathbb{P}}

\newcommand{\cA}{\mathcal{A}}

\newcommand{\cC}{\mathcal{C}}

\newcommand{\cY}{\mathcal{Y}}

\DeclareMathOperator*{\argmax}{argmax}

\DeclareMathOperator{\Polar}{Polar}

\DeclareMathOperator{\SER}{SER}

\DeclareMathOperator{\qSC}{qSC}
\DeclareMathOperator{\qEC}{qEC}
\begin{document}

\title{All the codeword symbols in polar codes have the same SER under the SC decoder}

\author{Guodong Li \and  \hspace*{.5in} Min Ye  \and  \hspace*{.5in}  Sihuang Hu}

\maketitle
{\renewcommand{\thefootnote}{}\footnotetext{

		\vspace{-.2in}

		\noindent\rule{1.5in}{.4pt}

		Research partially funded by National Key R\&D Program of China under Grant No. 2021YFA1001000, National Natural Science Foundation of China under Grant No. 12001322, Shandong Provincial Natural Science Foundation under Grant No. ZR202010220025, and a Taishan scholar program of Shandong Province.

		Guodong Li is with School of Cyber Science and Technology, Shandong University, Qingdao, Shandong, 266237, China.
		Email: guodongli@mail.sdu.edu.cn

		Min Ye is with Tsinghua-Berkeley Shenzhen Institute, Tsinghua Shenzhen International Graduate School, Shenzhen 518055, China.
		Email: yeemmi@gmail.com

		Sihuang Hu is with  Key Laboratory of Cryptologic Technology and Information Security, Ministry of Education, Shandong University, Qingdao, Shandong, 266237, China and School of Cyber Science and Technology, Shandong University, Qingdao, Shandong, 266237, China.
		Email: husihuang@sdu.edu.cn
	}
}

\renewcommand{\thefootnote}{\arabic{footnote}}
\setcounter{footnote}{0}

\begin{abstract}
We consider polar codes constructed from the $2\times 2$ kernel $\begin{bmatrix}
		1      & 0 \\
		\alpha & 1
	\end{bmatrix}$ over a finite field $\Fq$, where $q=p^s$ is a power of a prime number $p$, and $\alpha$ satisfies that $\ff_p(\alpha) = \Fq$.
We prove that for any $\Fq$-symmetric memoryless channel, any code length, and any code dimension, all the codeword symbols in such polar codes have the same symbol error rate (SER) under the successive cancellation (SC) decoder.
\end{abstract}

\section{Introduction}
Polar codes were proposed by Ar{\i}kan in \cite{Arikan09}, and they have been extensively studied over the last decade.
In his original paper \cite{Arikan09}, Ar{\i}kan introduced the successive cancellation (SC) decoder to decode polar codes, and he proved that polar codes achieve the capacity of any binary-input memoryless symmetric (BMS) channel under the SC decoder.
Later, the successive cancellation list (SCL) decoder and the CRC-aided SCL decoder were proposed to further reduce the decoding error probability of polar codes \cite{Tal15,Niu12}.

Although the CRC-aided SCL decoder provides state-of-the-art performance in terms of the decoding error probability, the SC decoder still receives a lot of research attention due to the following two reasons.
First, it is amenable to theoretical analysis.
In fact, a large part of theoretical research on polar codes focuses on the performance of the SC decoder.
Second, the running time of the SC decoder is much smaller than the (CRC-aided) SCL decoder.
Therefore, SC decoder is the better choice when there is a stringent requirement on the delay of the communication system.

In an early paper \cite{Arikan11}, Ar{\i}kan observed an interesting experimental result regarding the decoding performance of the SC decoder for binary polar codes: The average bit error rate (BER) of a subset of codeword coordinates is much smaller than the average BER of the message bits.
Even till today, there is no rigorous analysis that can explain this phenomenon.
As we repeated the experiments in \cite{Arikan11}, we found another interesting phenomenon---the BER of each codeword coordinates in polar codes is surprisingly stable under the SC decoder.
This is not a coincidence. In fact, Theorem~2 of \cite{Geiselhart21} implies that {\em all the codeword coordinates in binary polar codes have the same BER under the SC decoder}, although this conclusion was never explicitly mentioned in \cite{Geiselhart21} or any other papers.

In this paper, we extend the above conclusion to polar codes over finite fields. Specifically, we consider polar codes constructed from the $2\times 2$ kernel $\begin{bmatrix}
		1      & 0 \\
		\alpha & 1
\end{bmatrix}$ over a finite field $\Fq$, where $q=p^s$ is a power of a prime number $p$, and $\alpha$ satisfies that $\ff_p(\alpha) = \Fq$. According to \cite[Corollary 15]{Mori14},
$\ff_p(\alpha) = \Fq$ is a necessary and sufficient condition for the constructed codes to polarize.
We prove that for any $\ff_q$-symmetric memoryless channel, any code length, and any code dimension, all the codeword symbols in polar codes have exactly the same symbol error rate (SER) under the SC decoder.

The rest of this paper is organized as follows.
In Section~\ref{sect:2}, we provide the background on polar codes and state the main result.
In Section~\ref{sect:3}, we prove the main result.
In Section~\ref{sect:4}, we discuss the connection between our results and the observation in \cite{Arikan11}.

\section{Background and the main result} \label{sect:2}

For a set $\cA=\{i_1,i_2,\dots,i_s\}$ of size $s$, we use $x_{\cA}$ to denote the vector $(x_{i_1},x_{i_2},\dots,x_{i_s})$, where we assume that $i_1<i_2<\dots<i_s$ are nonnegative integers.
We use $x_{[a:b]}$ to denote the vector $(x_a,x_{a+1},x_{a+2},\dots,x_b)$.
Let $q = p^s$, where $p$ is a prime number and $s$ is a positive integer. We say that a memoryless channel $W:\Fq\to\cY$ is $\Fq$-symmetric if it satisfies the following two conditions: (1) there exist $q$ permutations $\{\sigma_b: b\in\Fq\}$ on the output alphabet $\cY$ such that $W[y|x] = W[\sigma_{x'-x}(y)|x']$ for all $y\in\cY$ and all $x,x'\in\Fq$;
(2) there exist $q-1$ permutations $\{\pi_a: a\in\Fq^*\}$ on $\cY$ such that $W[y|x] = W[\pi_a(y)|a x]$ for all $y\in \cY$, $x\in\Fq$, and $a\in\Fq^*$. Both $\qSC$ and $\qEC$ satisfy these two conditions. Below we use the shorthand notation 
\begin{equation} \label{eq:sth}
y+b=\sigma_b(y) \text{~~and~~}
a\cdot y=\pi_a(y)
\quad \text{for~} 
a\in\Fq^* ,  b\in\Fq, y\in\cY.
\end{equation}

The construction of polar codes with code length $n=2^m$ involves the $n\times n$ matrix
$\mathbf{G}_n=\begin{bmatrix}
		1      & 0 \\
		\alpha & 1
	\end{bmatrix}^{\otimes m}$,
where $\alpha$ satisfies that $\ff_p(\alpha) = \Fq$, and $\otimes$ is the Kronecker product. According to \cite[Corollary 15]{Mori14},
$\ff_p(\alpha) = \Fq$ is a necessary and sufficient condition for the polarization of the constructed codes.
The matrix $\mathbf{G}_n$ serves as a linear mapping between the message vector $u_{[0:n-1]}$ and the codeword vector $x_{[0:n-1]}$.
The message vector $u_{[0:n-1]}$ consists of information symbols and frozen symbols.
We use $\cA\subseteq\{0,1,\dots,n-1\}$ to denote the index set of information symbols and use $\cA^c=\{0,1,\dots,n-1\}\setminus\cA$ to denote the index set of frozen symbols.
A polar code has four parameters---the code length $n$, the code dimension $k$, the index set $\cA$ of information symbols, and the vector $\bar{u}_{\cA^c}\in\Fq^{n-k}$ of frozen symbols\footnote{We use $\bar{u}$ and its variations to denote the true value.
	We use $\hat{u}$ and its variations to denote the decoded value.}.
More precisely, we define
$$
	\Polar(n,k,\cA,\bar{u}_{\cA^c}) := \{u_{[0:n-1]} \mathbf{G}_n : u_{\cA^c}=\bar{u}_{\cA^c} , u_{\cA}\in\Fq^k \} .
$$

Next let us recall how the SC decoder works.
For a symmetric channel $W:\Fq\to\cY$, we define $W^n:\Fq^n\to\cY^n$ as $W^n(y_{[0:n-1]}|x_{[0:n-1]})=\prod_{i=0}^{n-1}W(y_i|x_i)$ for $x_{[0:n-1]}\in\Fq^n$ and $y_{[0:n-1]}\in\cY^n$.
For $0\le i\le n-1$, we further define the synthetic channel $W_i^{(n)}:\Fq\to\cY^n\times\Fq^i$ as
$$
	W_i^{(n)}(y_{[0:n-1]}, u_{[0:i-1]} |u_i)
	=\frac{1}{q^{n-1}} \sum_{u_{[i+1:n-1]}\in\Fq^{n-i-1}}
	W^n(y_{[0:n-1]}| u_{[0:n-1]} \mathbf{G}_n) .
$$

The SC decoder decodes one by one from $u_0$ to $u_{n-1}$.
If $i\in\cA^c$, then the decoding result of the $i$th symbol is $\hat{u}_i=\bar{u}_i$.
If $i\in\cA$, then the decoding result of the $i$th symbol is
\begin{equation} \label{eq:SCinfo}
	\hat{u}_i= \hat{u}_i(y_{[0:n-1]}, u_{[0:i-1]})
	=\argmax_{u_i\in\Fq} W_i^{(n)}(y_{[0:n-1]}, u_{[0:i-1]} |u_i) .
\end{equation}
We write the elements in $\Fq$ as $\{a_0, a_1, a_2, \dots, a_{q-1}\}$.
If there is a tie, i.e., if
\begin{equation*}
	\begin{aligned}
		  & W_i^{(n)}(y_{[0:n-1]}, u_{[0:i-1]} |a_{i_0})     =\cdots=   & W_i^{(n)}(y_{[0:n-1]}, u_{[0:i-1]} |a_{i_{s-1}})  \\
		> & W_i^{(n)}(y_{[0:n-1]}, u_{[0:i-1]} |a_{i_{s}}) \ge\cdots\ge & W_i^{(n)}(y_{[0:n-1]}, u_{[0:i-1]} |a_{i_{q-1}}),
	\end{aligned}
\end{equation*}
where $\{i_0, i_1, \dots, i_{q-1}\}$ is a permutation of $\{0,1,\dots, q-1\}$,
then the SC decoder outputs a random decoding result with probability $\mathbb{P}(\hat{u}_i=a_{i_0})=\cdots=\mathbb{P}(\hat{u}_i=a_{i_{s-1}})=1/s$.
The decoding results of the SC decoder depend not only on the channel output vector $y_{[0:n-1]}$, but also on the set $\cA$ and the values of the frozen symbols $\bar{u}_{\cA^c}$.
We write the SC decoding result of the message vector as $\hat{u}_{[0:n-1]}(y_{[0:n-1]},\cA,\bar{u}_{\cA^c})$.
The SC decoding result of the codeword vector is
\begin{equation} \label{eq:decw}
	\hat{x}_{[0:n-1]}(y_{[0:n-1]},\cA,\bar{u}_{\cA^c})=\hat{u}_{[0:n-1]}(y_{[0:n-1]},\cA,\bar{u}_{\cA^c})\mathbf{G}_n .
\end{equation}

For a polar code $\cC=\Polar(n,k,\cA,\bar{u}_{\cA^c})$, a symmetric channel $W$, and a specific choice of information vector $\bar{u}_{\cA}\in\Fq^k$, we write the transmitted codeword as $\bar{x}_{[0:n-1]}=\bar{u}_{[0:n-1]} \mathbf{G}_n$.
In this case, the SER of the $j$th codeword symbol under the SC decoder is
\begin{align*}
	  & \SER_j(\cC,W,\bar{u}_{\cA})                                                                                                                                                                           \\
	= & \sum_{y_{[0:n-1]}\in\cY^n} ~\sum_{x_{[0:n-1]}\in\Fq^n} W^n(y_{[0:n-1]}| \bar{x}_{[0:n-1]}) \mathbb{P}(\hat{x}_{[0:n-1]}(y_{[0:n-1]},\cA,\bar{u}_{\cA^c})=x_{[0:n-1]}) \mathbf{1}[x_j\neq \bar{x}_j] .
\end{align*}
The term $\mathbb{P}(\hat{x}_{[0:n-1]}(y_{[0:n-1]},\cA,\bar{u}_{\cA^c})=x_{[0:n-1]})$ appears because the SC decoder involves randomness when there is a tie in \eqref{eq:SCinfo}.
The average SER of the $j$th codeword symbol over all choices of the information vector is
$$
	\SER_j(\cC,W) = \frac{1}{2^k} \sum_{\bar{u}_{\cA}\in\Fq^k}
	\SER_j(\cC,W,\bar{u}_{\cA}) .
$$

For an integer $0\le i\le 2^m-1$, we write its binary expansion as
\begin{equation} \label{eq:bexp}
	i= 2^{m-1}b_{m-1}(i) + 2^{m-2}b_{m-2}(i) +\dots+2b_1(i) + b_0(i),
\end{equation}
where $b_{m-1}(i),b_{m-2}(i),\dots,b_1(i),b_0(i)\in\{0,1\}$.
For two integers $0\le i,j\le 2^m-1$, we say that $i\succeq j$ if $b_r(i)\ge b_r(j)$ for all $0\le r\le m-1$.
Now we are ready to state our main result.
\begin{theorem} \label{thm:main}
	Let $n=2^m$ and $k$ be two positive integers satisfying $k\le n$.
	Let $\cA\subseteq\{0,1,\dots,n-1\}$ be a set of size $|\cA|=k$ satisfying the following condition:
	\begin{equation} \label{eq:condA}
		\text{If } j\in\cA \text{ and } i\succeq j, \text{ then } i\in\cA .
	\end{equation}
	Let $\bar{u}_{\cA^c}\in\Fq^{n-k}$ be a $q$-ary vector of length $n-k$.
	We write $\cC=\Polar(n,k,\cA,\bar{u}_{\cA^c})$.
	Then for any $\ff_q$-symmetric memoryless channel $W$, we have
	$$
		\SER_0(\cC,W)=\SER_1(\cC,W)=\SER_2(\cC,W)=\dots=\SER_{n-1}(\cC,W) .
	$$
\end{theorem}

We need the following lemma to apply Theorem~\ref{thm:main} to polar codes.

\begin{lemma}
	The set $\cA$ in polar code construction always satisfies the condition \eqref{eq:condA}.
\end{lemma}

This lemma was proved in \cite{bardet2016_arxiv, bardet2016_ISIT, schurch2016partial} for binary polar codes, and the extension to polar codes over finite fields is trivial. In fact, \cite{bardet2016_arxiv, bardet2016_ISIT, schurch2016partial} proved an even stronger result, showing that binary polar codes are decreasing monomial codes. The properties of decreasing monomial codes were used to reduce the complexity of encoding, decoding, and code construction for polar codes \cite{sarkis2016flexible, mondelli2019construction}.

Note that the condition \eqref{eq:condA} is necessary for Theorem~\ref{thm:main} to hold.
Let us consider the simplest case of $n=2,k=1$, and $q=2$.
The choice of $\cA=\{1\}$ satisfies the condition \eqref{eq:condA}.
Under this choice, $u_0$ is the frozen symbol, so its decoding result $\hat{u}_0$ is equal to the true value $\bar{u}_0$.
The decoding results of the codeword symbols are $\hat{x}_0=\hat{u}_0+\alpha\hat{u}_1$ and $\hat{x}_1=\hat{u}_1$.
Therefore, the three events $\{\hat{x}_0\neq \bar{x}_0\}$, $\{\hat{x}_1\neq \bar{x}_1\}$ and $\{\hat{u}_1\neq \bar{u}_1\}$ are the same, so the two codeword symbols have the same SER.
On the other hand, the choice of $\cA=\{0\}$ does not satisfy the condition \eqref{eq:condA}.
Under this choice, $u_1$ is the frozen symbol, so its decoding result $\hat{u}_1$ is equal to the true value $\bar{u}_1$, which implies that $\hat{x}_1=\bar{x}_1$.
However, the event $\{\hat{x}_0\neq \bar{x}_0\}$ happens with positive probability unless the channel $W$ is noiseless.
Therefore, the two codeword symbols do not have the same SER when the condition \eqref{eq:condA} is not satisfied.

\section{Proof of the main result} \label{sect:3}

\subsection{Restricting to the all-zero codeword} \label{sect:allzero}

By \eqref{eq:sth}, we have $W(y|x)=W(a\cdot y + b|a\cdot x + b)$ for all $a\in\Fq^*, b\in\Fq, x\in\Fq, y\in\cY$.
For $a\in \Fq^*$ and two vectors $b_{[0:n-1]}\in\Fq^n, y_{[0:n-1]}\in\cY^n$, we define
$a \cdot y_{[0:n-1]} + b_{[0:n-1]} =(a\cdot y_0 + b_0, a\cdot y_1 + b_1,\dots, a\cdot y_{n-1} + b_{n-1})$. 
Then we have $W^n(y_{[0:n-1]}|x_{[0:n-1]})=W^n(a\cdot y_{[0:n-1]} + b_{[0:n-1]}|a\cdot x_{[0:n-1]} + b_{[0:n-1]})$ for all $a\in \Fq^*$ and $b_{[0:n-1]}, x_{[0:n-1]}\in\Fq^n$.

We first prove that the SER of each codeword symbol is independent of the true value $\bar{u}_{[0:n-1]}$ of the message vector.
Note that the proof technique of Lemma~\ref{lm:2} below is quite standard in polar coding literature.
In fact, this proof technique was already used in Ar{\i}kan's original paper; see Section~VI of \cite{Arikan09}.
The only difference is that the proof in \cite{Arikan09} focused on the word error rate while we focus on SER.

\begin{lemma} \label{lm:2}
	For any $\bar{u}_{[0:n-1]}\in\Fq^n$, any $\bar{u}'_{[0:n-1]}\in\Fq^n$, and any $0\le j\le n-1$, we have
	$$
		\SER_j(\Polar(n,k,\cA,\bar{u}_{\cA^c}),W,\bar{u}_{\cA}) = \SER_j(\Polar(n,k,\cA,\bar{u}'_{\cA^c}),W,\bar{u}'_{\cA}) .
	$$
\end{lemma}

\begin{proof}
For any $\bar{u}_{[0:n-1]}\in\Fq^n$ and $\bar{u}'_{[0:n-1]}\in\Fq^n$, we can always find a pair $a\in\Fq^*$ and $b_{[0:n-1]}\in\Fq^n$ such that
$$
\bar{u}'_{[0:n-1]} = a \cdot \bar{u}_{[0:n-1]} + b_{[0:n-1]} .
$$
In fact, for each choice of $a\in \Fq^*$, we can always find a vector $b_{[0:n-1]}\in\Fq^n$ satisfying the above equation, so there are $q-1$ choices of the pair $(a, b_{[0:n-1]})$ in total. We may pick an arbitrary one of them in this proof\footnote{For the proof of Lemma~\ref{lm:2}, we can simply choose $a=1$. Other choices of $a$ are needed for the proof of Lemma~\ref{lm:3}.}.

Define $\bar{x}_{[0:n-1]}=\bar{u}_{[0:n-1]} \mathbf{G}_n$,
	$\bar{x}'_{[0:n-1]}=\bar{u}'_{[0:n-1]} \mathbf{G}_n$,
	and $x^b_{[0:n-1]} = b_{[0:n-1]}\cdot \mathbf{G}_n$.
Then we have $\bar{x}'_{[0:n-1]} = a \cdot \bar{x}_{[0:n-1]} + x^b_{[0:n-1]}$, so
	\begin{equation} \label{eq:sp0}
		\begin{aligned}
			 & W^n(y_{[0:n-1]}| \bar{x}_{[0:n-1]} )
			= W^n(a \cdot y_{[0:n-1]} + x^b_{[0:n-1]}| \bar{x}'_{[0:n-1]} ) \quad \text{for all~} y_{[0:n-1]}\in\cY^n ,                          \\
			 & \mathbf{1}[z\neq \bar{x}_j]=\mathbf{1}[a\cdot z + x_j^b \neq \bar{x}'_j]  \quad  \text{for all~} z\in\Fq \text{~and all~} 0\le j\le n-1 .
		\end{aligned}
	\end{equation}
	In this proof, we use the shorthand notation
	\begin{equation} \label{eq:shd1}
		\hat{u}_{[0:n-1]}(y_{[0:n-1]})=\hat{u}_{[0:n-1]}(y_{[0:n-1]},\cA,\bar{u}_{\cA^c}) , \quad \hat{u}'_{[0:n-1]}(y_{[0:n-1]})=\hat{u}_{[0:n-1]}(y_{[0:n-1]},\cA,\bar{u}'_{\cA^c}) .
	\end{equation}
	We first prove that
	\begin{equation} \label{eq:sp1}
		\mathbb{P} \big(\hat{u}_{[0:n-1]}(y_{[0:n-1]})=u_{[0:n-1]} \big)=
		\mathbb{P} \big(\hat{u}'_{[0:n-1]}(a\cdot y_{[0:n-1]} + x^b_{[0:n-1]})=a\cdot u_{[0:n-1]} +b_{[0:n-1]} \big)
	\end{equation}
	for all $u_{[0:n-1]}\in\Fq^n$.
	The randomness here comes from the random decision of the SC decoder when there is a tie in \eqref{eq:SCinfo}.
	Note that
	\begin{align*}
		  & \mathbb{P} \big(\hat{u}_{[0:n-1]}(y_{[0:n-1]})=u_{[0:n-1]} \big) = \prod_{i=0}^{n-1} \mathbb{P} \big( \hat{u}_i(y_{[0:n-1]}, u_{[0:i-1]}) = u_i \big) ,               \\
		  & \mathbb{P} \big(\hat{u}'_{[0:n-1]}(a \cdot y_{[0:n-1]} + x^b_{[0:n-1]})=a \cdot u_{[0:n-1]} +b_{[0:n-1]} \big)                                    \\
		= & \prod_{i=0}^{n-1} \mathbb{P} \big(\hat{u}'_i(  a \cdot y_{[0:n-1]} + x^b_{[0:n-1]}, a \cdot u_{[0:i-1]} + b_{[0:i-1]} )=a \cdot u_i + b_i\big) .
	\end{align*}
	Therefore, in order to prove \eqref{eq:sp1}, we only need to show that
	\begin{equation} \label{eq:sp2}
		\mathbb{P} \big( \hat{u}_i(y_{[0:n-1]}, u_{[0:i-1]}) = u_i \big)
		= \mathbb{P} \big(\hat{u}'_i(  a \cdot y_{[0:n-1]} + x^b_{[0:n-1]}, a \cdot u_{[0:i-1]} + b_{[0:i-1]} )=a \cdot u_i + b_i\big)
	\end{equation}
	for all $0\le i\le n-1$ and all $u_{[0:i]}\in\Fq^{i+1}$.
	If $i\in\cA^c$, then $\hat{u}_i=\bar{u}_i$ and $\hat{u}'_i=\bar{u}'_i$.
	Since $\bar{u}'_i=a \cdot \bar{u}_i +b_i$, the equality \eqref{eq:sp2} clearly holds.
	If $i\in\cA$, then we need to analyze the synthetic channel in \eqref{eq:SCinfo}.
	More specifically, we have
	\begin{align*}
		  & W_i^{(n)}(y_{[0:n-1]}, u_{[0:i-1]} |u_i)
		=\frac{1}{q^{n-1}} \sum_{u_{[i+1:n-1]}\in\Fq^{n-i-1}}
		W^n(y_{[0:n-1]}| u_{[0:n-1]} \mathbf{G}_n)                                                                                     \\
		= & \frac{1}{q^{n-1}} \sum_{u_{[i+1:n-1]}\in\Fq^{n-i-1}}
		W^n(a \cdot y_{[0:n-1]} + x^b_{[0:n-1]}| (a \cdot u_{[0:n-1]} +b_{[0:n-1]}) \mathbf{G}_n)                   \\
		= & W_i^{(n)} (a \cdot y_{[0:n-1]} + x^b_{[0:n-1]}, a \cdot u_{[0:i-1]} +b_{[0:i-1]} | a \cdot u_i +b_i) .
	\end{align*}
	The last equality holds because when $u_{[i+1:n-1]}$ ranges over all values in $\Fq^{n-i-1}$, the sum $a \cdot u_{[i+1:n-1]}+b_{[i+1:n-1]}$ also ranges over all values in $\Fq^{n-i-1}$.
	This equality together with the decoding rule \eqref{eq:SCinfo} immediately implies that \eqref{eq:sp2} holds for $i\in\cA$.
	This completes the proof of \eqref{eq:sp1} and \eqref{eq:sp2}.

	Recall the notation for the decoding result of the codeword vector in \eqref{eq:decw}.
	Similarly to \eqref{eq:shd1}, we use the shorthand notation
	$$
		\hat{x}_{[0:n-1]}(y_{[0:n-1]})=\hat{x}_{[0:n-1]}(y_{[0:n-1]},\cA,\bar{u}_{\cA^c}) , \quad \hat{x}'_{[0:n-1]}(y_{[0:n-1]})=\hat{x}_{[0:n-1]}(y_{[0:n-1]},\cA,\bar{u}'_{\cA^c}) .
	$$
	With this notation, Equation~\eqref{eq:sp1} is equivalent to
	\begin{equation} \label{eq:frx}
		\mathbb{P} \big(\hat{x}_{[0:n-1]}(y_{[0:n-1]})=x_{[0:n-1]} \big)=
		\mathbb{P} \big(\hat{x}'_{[0:n-1]}(a \cdot y_{[0:n-1]} + x^b_{[0:n-1]})= a \cdot x_{[0:n-1]} + x^b_{[0:n-1]} \big)
	\end{equation}
	for all $x_{[0:n-1]}\in\Fq^n$.
	Combining this with \eqref{eq:sp0}, we have
	\begin{align*}
		                 & \SER_j(\Polar(n,k,\cA,\bar{u}_{\cA^c}),W,\bar{u}_{\cA})                                                                                                                                                             \\
		=                & \sum_{y_{[0:n-1]}\in\cY^n} ~\sum_{x_{[0:n-1]}\in\Fq^n} W^n(y_{[0:n-1]}| \bar{x}_{[0:n-1]}) \mathbb{P}(\hat{x}_{[0:n-1]}(y_{[0:n-1]})=x_{[0:n-1]}) \mathbf{1}[x_j\neq \bar{x}_j]                                     \\
		=                & \sum_{y_{[0:n-1]}\in\cY^n} ~\sum_{x_{[0:n-1]}\in\Fq^n} \Big( W^n(a  \cdot y_{[0:n-1]} + x^b_{[0:n-1]}| \bar{x}'_{[0:n-1]} )                                                                                \\
		                 & \hspace*{0.4in} \times \mathbb{P} \big(\hat{x}'_{[0:n-1]}(a  \cdot y_{[0:n-1]} + x^b_{[0:n-1]})=a \cdot x_{[0:n-1]} + x_{[0:n-1]}^b \big) \mathbf{1}[a \cdot x_j + x_j^b \neq \bar{x}'_j] \Big) \\
		\overset{(*)}{=} & \sum_{y_{[0:n-1]}\in\cY^n} ~\sum_{x_{[0:n-1]}\in\Fq^n} W^n(y_{[0:n-1]}| \bar{x}'_{[0:n-1]}) \mathbb{P}(\hat{x}'_{[0:n-1]}(y_{[0:n-1]})=x_{[0:n-1]}) \mathbf{1}[x_j\neq \bar{x}'_j]                                  \\
		=                & \SER_j(\Polar(n,k,\cA,\bar{u}'_{\cA^c}),W,\bar{u}'_{\cA}) .
	\end{align*}
	The equality $(*)$ is obtained from replacing $a \cdot y_{[0:n-1]} + x^b_{[0:n-1]}$ with $y_{[0:n-1]}$ and replacing $a \cdot x_{[0:n-1]} + x_{[0:n-1]}^b$ with $x_{[0:n-1]}$.
	These two replacements are eligible because when $y_{[0:n-1]}$ ranges over $\cY^n$, $a \cdot y_{[0:n-1]}+x^b_{[0:n-1]}$ also ranges over all values in $\cY^n$; similarly, when $x_{[0:n-1]}$ ranges over $\Fq^n$, $a \cdot x_{[0:n-1]} + x^b_{[0:n-1]}$ also ranges over all values in $\Fq^n$.
	Moreover, since $a \cdot x_j + x_j^b$ is the $j$th coordinate of $a \cdot x_{[0:n-1]} \oplus x_{[0:n-1]}^b$, we also replace $a \cdot x_j+x_j^b$ with $x_j$ when we replace $a \cdot x_{[0:n-1]} + x^b_{[0:n-1]}$ with $x_{[0:n-1]}$.
	This completes the proof of the lemma.
\end{proof}

This lemma implies that in the analysis of SER of the SC decoder, we can always assume that the true value $\bar{u}_{[0:n-1]}$ of the message vector (including both information symbols and frozen symbols) is all-zero, or equivalently, the transmitted codeword is all-zero.
More specifically, we have the following expression for the SER:
\begin{equation} \label{eq:uto0}
	\begin{aligned}
		  & \SER_j(\Polar(n,k,\cA,\bar{u}_{\cA^c}),W)                                                                                                                               \\
		= & \SER_j(\Polar(n,k,\cA,0^{n-k}),W)
		= \SER_j(\Polar(n,k,\cA,0^{n-k}),W, 0^k)                                                                                                                                    \\
		= & \sum_{y_{[0:n-1]}\in\cY^n} ~\sum_{x_{[0:n-1]}\in\Fq^n} W^n(y_{[0:n-1]}| 0^n) \mathbb{P}(\hat{x}_{[0:n-1]}(y_{[0:n-1]},\cA,0^{n-k})=x_{[0:n-1]}) \mathbf{1}[x_j\neq 0] ,
	\end{aligned}
\end{equation}
where $0^i$ is the all-zero vector of length $i$.

Note that \eqref{eq:frx} immediately implies the following lemma, which will be used later to prove the main theorem.
\begin{lemma} \label{lm:3}
	Let $x^b_{[0:n-1]}=b_{[0:n-1]} \mathbf{G}_n$ be a codeword of $\Polar(n,k,\cA,0^{n-k})$, i.e., $b_i=0$ for all $i\in\cA^c$.
	Then
	\begin{equation*}
		\begin{aligned}
			  & \mathbb{P} \big(\hat{x}_{[0:n-1]}(y_{[0:n-1]},\cA,0^{n-k})=x_{[0:n-1]} \big)                                                                  \\
			= & \mathbb{P} \big(\hat{x}_{[0:n-1]}(a \cdot y_{[0:n-1]}+x^b_{[0:n-1]},\cA,0^{n-k})= a\cdot x_{[0:n-1]} + x_{[0:n-1]}^b \big)
		\end{aligned}
	\end{equation*}
	for all $a \in \Fq^*, x_{[0:n-1]}\in\Fq^n$ and all $y_{[0:n-1]}\in\cY^n$.
\end{lemma}
As a final remark, we note that we do not need condition \eqref{eq:condA} for the set $\cA$ in Lemma~\ref{lm:2} and Lemma~\ref{lm:3}, but we will need this condition later in the proof of the main theorem.

\subsection{Recursive implementation of the SC decoder}

In practice, the SC decoder is usually implemented recursively based on the transition probability.
Below we recap this recursive structure since it is needed in the next step of our proof.
Recall that $W: \Fq\to \cY$ is a symmetric channel.
We denote the set of transition probabilities for an output symbol $y\in\cY$ as
$$
	T(y) = \{W(y|u): u\in\Fq\}.
$$
For a channel output vector $y_{[0:n-1]}\in\cY^n$, we define a vector
$$
	T_{[0:n-1]} = T_{[0:n-1]}(y_{[0:n-1]}) = (T_0, T_1, \dots, T_{n-1}), \text{~where~} T_i = T(y_i) \text{~for~} 0\le i \le n-1.
$$
Note that each coordinate $T_i = T(y_i)$ is the set of transition probabilities for the output symbol $y_i$.
We write $T_{[0:n-1]}(y_{[0:n-1]})$ when we want to emphasize its dependence on $y_{[0:n-1]}$.
In most scenarios this dependence is clear from the context, and we will simply write $T_{[0:n-1]}$.

It is well known that the decoding result of the SC decoder only depends on the transition probabilities of the channel outputs.
In other words, if two channel output vectors have the same vector of transition probabilitiy sets, then their decoding results have the same probability distribution\footnote{Recall that the SC decoder produces a random decoding result when there is a tie in \eqref{eq:SCinfo}.} under the SC decoder.
Below we will write $\hat{u}_i(y_{[0:n-1]}, u_{[0:i-1]})$ and $\hat{u}_i(T_{[0:n-1]}, u_{[0:i-1]})$ interchangeably.
We also write $\hat{u}_{[0:n-1]}(y_{[0:n-1]},\cA,\bar{u}_{\cA^c})$ and $\hat{u}_{[0:n-1]}(T_{[0:n-1]},\cA,\bar{u}_{\cA^c})$ interchangeably.

For $a \in\Fq^*, b_{[0:n-1]}\in\Fq^n$, and $T_{[0:n-1]}=T_{[0:n-1]}(y_{[0:n-1]})$,
we define
\begin{align*}
	a \cdot T_{[0:n-1]} + b_{[0:n-1]}
	=(a \cdot T_0+b_0, a \cdot T_1+b_1,\dots, a \cdot T_{n-1}+b_{n-1}), \\
	\text{where~} a \cdot T_i+b_i=T(a \cdot y_i+b_i)
	\text{~for~} 0\le i \le n-1.
\end{align*}
Each coordinate $a \cdot T_i+b_i=T(a \cdot y_i+b_i)$ is the set of transition probabilities for the output symbol $a \cdot y_i+b_i$.
With the new notation, we restate Lemma~\ref{lm:3} as follows.
\begin{lemma} \label{lm:4}
	Let $x^b_{[0:n-1]}=b_{[0:n-1]} \mathbf{G}_n$ be a codeword of $\Polar(n,k,\cA,0^{n-k})$, i.e., $b_i=0$ for all $i\in\cA^c$.
	Then
	\begin{equation*}
		\begin{aligned}
			  & \mathbb{P} \big(\hat{x}_{[0:n-1]}(T_{[0:n-1]},\cA,0^{n-k})=x_{[0:n-1]} \big)                                                                  \\
			= & \mathbb{P} \big(\hat{x}_{[0:n-1]}(a \cdot T_{[0:n-1]}+x^b_{[0:n-1]},\cA,0^{n-k})=a \cdot x_{[0:n-1]} + x_{[0:n-1]}^b \big)
		\end{aligned}
	\end{equation*}
	for all $a \in \Fq^*, x_{[0:n-1]}\in\Fq^n$ and all $y_{[0:n-1]}\in\cY^n$.
\end{lemma}

We need some more notation to describe the recursive structure of the SC decoder.
For $y_0,y_1\in\cY$ and $u_0\in\Fq$, we define two sets $T^-(y_0, y_1) = \{W^-(y_0,y_1|u): u\in\Fq\}$ and $T^+(y_0, y_1, u_0) = \{W^+(y_0,y_1, u_0|u): u\in\Fq\}$  of size $q$, where
\begin{equation} \label{eq:Lpm}
	\begin{aligned}
		 & W^-(y_0,y_1     | u)  = \frac 1 q \sum_{u_1\in\Fq} W(y_0|u + \alpha u_1)W(y_1| u_1), \\
		 & W^+(y_0,y_1,u_0 | u)  = \frac 1 q W(y_0|u_0 + \alpha u)W(y_1| u) .
	\end{aligned}
\end{equation}
For $y_{[0:n-1]}\in\cY^n$ and $u_{[0:n/2-1]}\in\Fq^{n/2}$, we define
\begin{equation} \label{eq:Lvecpm}
	\begin{aligned}
		 & T_{[0:n/2-1]}^-(y_{[0:n-1]})=(T_0^-, T_1^-,\dots,T_{n/2-1}^-) , \quad \text{where~} T_i^-= T^-(y_i,y_{i+n/2}) \text{~for~} 0\le i\le n/2-1 , \\
		 & T_{[0:n/2-1]}^+(y_{[0:n-1]},u_{[0:n/2-1]})=(T_0^+, T_1^+,\dots,T_{n/2-1}^+) ,                                                                \\
		 & \hspace*{2.5in} \text{where~} T_i^+= T^+(y_i,y_{i+n/2},u_i) \text{~for~} 0\le i\le n/2-1 .
	\end{aligned}
\end{equation}
For a set $\cA\subseteq\{0,1,\dots,n-1\}$, we define
$$
	\cA^-=\cA\cap \{0,1,\dots,n/2-1\} \text{~~and~~}
	\cA^+=\{i:i+n/2\in\cA\cap\{n/2,n/2+1,\dots,n-1\}\} .
$$

The encoding procedure $x_{[0:n-1]}=u_{[0:n-1]} \mathbf{G}_n$ for polar codes can be decomposed in the following way: Let $z_{[0:n/2-1]}=u_{[0:n/2-1]} \mathbf{G}_{n/2}$ and $z_{[n/2:n-1]}=u_{[n/2:n-1]} \mathbf{G}_{n/2}$, i.e., we first encode the two halves separately.
Then $x_{[0:n/2-1]}=z_{[0:n/2-1]} + \alpha \cdot  z_{[n/2:n-1]}$ and $x_{[n/2:n-1]}=z_{[n/2:n-1]}$.
This recursive structure can be summarized as
\begin{equation} \label{eq:sfu}
	u_{[0:n-1]} \mathbf{G}_{n}=((u_{[0:n/2-1]} + \alpha \cdot  u_{[n/2:n-1]})\mathbf{G}_{n/2}, u_{[n/2:n-1]} \mathbf{G}_{n/2} ) .
\end{equation}

From this point on, we assume that all the frozen symbols take value $0$, which is justified by Lemma~\ref{lm:2}.
We use the shorthand notation
\begin{align*}
	\hat{u}_{[0:n-1]}(T_{[0:n-1]},\cA) = \hat{u}_{[0:n-1]}(T_{[0:n-1]},\cA,0^{n-k}) , \quad
	\hat{x}_{[0:n-1]}(T_{[0:n-1]},\cA) = \hat{x}_{[0:n-1]}(T_{[0:n-1]},\cA,0^{n-k}) .
\end{align*}
Given a channel output vector $y_{[0:n-1]}$, the SC decoder first decodes $u_{[0:n/2-1]}$ as
$$
	\hat{u}_{[0:n/2-1]}(T_{[0:n/2-1]}^-(y_{[0:n-1]}),\cA^-) .
$$
Then it calculates $\hat{z}_{[0:n/2-1]}=\hat{u}_{[0:n/2-1]}(T_{[0:n/2-1]}^-(y_{[0:n-1]}),\cA^-)\mathbf{G}_{n/2}$.
In the next step, the SC decoder decodes $u_{[n/2:n-1]}$ as
$$
	\hat{u}_{[0:n/2-1]}(T_{[0:n/2-1]}^+(y_{[0:n-1]},\hat{z}_{[0:n/2-1]}),\cA^+) .
$$
In summary, we have
\begin{equation} \label{eq:reSC}
	\begin{aligned}
		  & \bP\big(\hat{u}_{[0:n-1]}(y_{[0:n-1]},\cA) = u_{[0:n-1]} \big)                                                               \\
		= & \bP\big(\hat{u}_{[0:n/2-1]}(T_{[0:n/2-1]}^-(y_{[0:n-1]}),\cA^-) = u_{[0:n/2-1]} \big)                                        \\
		  & \times \bP\big( \hat{u}_{[0:n/2-1]}(T_{[0:n/2-1]}^+(y_{[0:n-1]},u_{[0:n/2-1]} \mathbf{G}_{n/2}),\cA^+) = u_{[n/2:n-1]} \big)
	\end{aligned}
\end{equation}
for all $y_{[0:n-1]}\in\cY^n$ and all $u_{[0:n-1]}\in\Fq^n$.
In light of \eqref{eq:sfu}, Equation~\eqref{eq:reSC} is equivalent to
\begin{equation} \label{eq:recx}
	\begin{aligned}
		  & \bP\big(\hat{x}_{[0:n-1]}(y_{[0:n-1]},\cA) = (z_{[0:n/2-1]} + \alpha \cdot z_{[n/2:n-1]}, z_{[n/2:n-1]}) \big) \\
		= & \bP\big(\hat{x}_{[0:n/2-1]}(T_{[0:n/2-1]}^-(y_{[0:n-1]}),\cA^-) = z_{[0:n/2-1]} \big)                          \\
		  & \times \bP\big( \hat{x}_{[0:n/2-1]}(T_{[0:n/2-1]}^+(y_{[0:n-1]}, z_{[0:n/2-1]} ),\cA^+) = z_{[n/2:n-1]} \big)
	\end{aligned}
\end{equation}
for all $y_{[0:n-1]}\in\cY^n$ and all $z_{[0:n-1]}\in\Fq^n$.

\subsection{Proof of Theorem~\ref{thm:main}}
For $n=2^m$, we define $m$ permutations $\delta_0^{(m)},\delta_1^{(m)},\dots,\delta_{m-1}^{(m)}$ on the set $\{0,1,\dots,n-1\}$.
For $0\le i\le n-1$, let $(b_{m-1}(i),b_{m-2}(i),\dots,b_1(i),b_0(i))$ be its binary expansion defined in \eqref{eq:bexp}.
For $0\le r\le m-1$, $\delta_r^{(m)}(i)$ is obtained from flipping the $r$th digit in the binary expansion of $i$.
In other words, the binary expansion of $\delta_r^{(m)}(i)$ is
$$
	(b_{m-1}(i),b_{m-2}(i),\dots, b_{r+1}(i), b_r(i)\oplus 1, b_{r-1}(i),\dots, b_1(i),b_0(i)) .
$$
As a concrete example, if we apply the permutation $\delta_1^{(2)}$ to each coordinate of $(0,1,2,3)$, then we obtain $(2,3,0,1)$; if we apply the permutation $\delta_0^{(2)}$ to each coordinate of $(0,1,2,3)$, then we obtain $(1,0,3,2)$.
We further define $m$ mappings $\xi_0^{(m)},\xi_1^{(m)},\dots,\xi_{m-1}^{(m)}$ on vectors of length $n$.
For $0\le r\le m-1$ and a length-$n$ vector $x_{[0:n-1]}$, we define
$$
	\xi_r^{(m)}(x_{[0:n-1]})=(a_0\cdot x_{\delta_r^{(m)}(0)}, a_1\cdot x_{\delta_r^{(m)}(1)},\dots,a_{n-1}\cdot x_{\delta_r^{(m)}(n-1)}),
$$
where $a_i = -\alpha$ if $b_r(i) = 0$ and $a_i = -\alpha^{-1}$ if $b_r(i) = 1$.
In particular, we have
$$
	\xi_{m-1}^{(m)}(x_{[0:n-1]})=(-\alpha\cdot x_{[n/2:n-1]}, -\alpha^{-1} \cdot x_{[0:n/2-1]}) .
$$

\begin{lemma} \label{lm:5}
	Let $n=2^m$.
	Suppose that all the frozen symbols take value $0$.
	Suppose that the index set $\cA$ of information symbols satisfies the condition \eqref{eq:condA}.
	Then
	\begin{equation}  \label{eq:r2}
		\bP\big( \hat{x}_{[0:n-1]}(y_{[0:n-1]},\cA) = x_{[0:n-1]}\big)
		= \bP\big(\hat{x}_{[0:n-1]}(\xi_{m-1}^{(m)}(y_{[0:n-1]}),\cA) = \xi_{m-1}^{(m)}(x_{[0:n-1]}) \big)
	\end{equation}
	for all $y_{[0:n-1]}\in\cY^n$ and all $x_{[0:n-1]}\in\Fq^n$.
\end{lemma}
\begin{proof}
	We will prove another equation that is equivalent to \eqref{eq:r2}.
	Specifically, we will prove that
	\begin{equation} \label{eq:ktg}
		\begin{aligned}
			  & \bP\big( \hat{x}_{[0:n-1]}(y_{[0:n-1]},\cA) = (z_{[0:n/2-1]} +\alpha \cdot  z_{[n/2:n-1]}, z_{[n/2:n-1]})\big)                                                                              \\
			= & \bP\Big( \hat{x}_{[0:n-1]}( \xi_{m-1}^{(m)}(y_{[0:n-1]}),\cA) = \big(-\alpha \cdot z_{[n/2:n-1]}, -\alpha^{-1}\cdot (z_{[0:n/2-1]} +\alpha \cdot  z_{[n/2:n-1]})\big) \Big)
		\end{aligned}
	\end{equation}
	for all $y_{[0:n-1]}\in\cY^n$ and all $z_{[0:n-1]}\in\Fq^n$.
These two equations are equivalent because	\eqref{eq:ktg} is obtained from replacing $x_{[0:n-1]}$ with $(z_{[0:n/2-1]} +\alpha \cdot z_{[n/2:n-1]}, z_{[n/2:n-1]})$ in \eqref{eq:r2}.

By \eqref{eq:Lpm}, we have
\begin{align*}
& W^-(-\alpha \cdot y_1, -\alpha^{-1}\cdot y_0 | u_0) = \frac 1 q \sum_{u_1\in\Fq} W(-\alpha \cdot y_1 |u_0 + \alpha u_1)W(-\alpha^{-1}\cdot y_0| u_1) \\
= & \frac 1 q \sum_{u_1\in\Fq} W( y_1 |-\alpha^{-1} u_0 - u_1) W(y_0|-\alpha u_1)
= \frac 1 q \sum_{u_1\in\Fq} W(y_0|-\alpha u_1) W( y_1 |-\alpha^{-1} u_0 - u_1) \\
\overset{(a)}{=} & \frac 1 q \sum_{v\in\Fq} W(y_0|u_0+ \alpha v) W( y_1 | v)
= W^-(y_0,y_1 | u_0) \quad \quad
\text{for all~} u_0\in\Fq \text{~and~} y_0,y_1\in\cY ,
\end{align*}
where equality $(a)$ is obtained by setting $v=-\alpha^{-1} u_0 - u_1$ and observing that $-\alpha u_1=u_0+ \alpha v$.
Therefore, $T^-(y_0,y_1) = T^-(-\alpha \cdot y_1, -\alpha^{-1}\cdot y_0)$.
Definition \eqref{eq:Lvecpm} further implies that $T_{[0:n/2-1]}^-(y_{[0:n-1]})=T_{[0:n/2-1]}^-( \xi_{m-1}^{(m)}(y_{[0:n-1]}))$ for all $y_{[0:n-1]}\in\cY^n$.
	Therefore,
	\begin{equation} \label{eq:00}
		\begin{aligned}
			  & \bP\big(\hat{x}_{[0:n/2-1]}(T_{[0:n/2-1]}^-(y_{[0:n-1]}),\cA^-) = z_{[0:n/2-1]} \big)                                     \\
			= & \bP\big(\hat{x}_{[0:n/2-1]}(T_{[0:n/2-1]}^-( \xi_{m-1}^{(m)}(y_{[0:n-1]})),\cA^-) = z_{[0:n/2-1]} \big)
		\end{aligned}
	\end{equation}
	for all $y_{[0:n-1]}\in\cY^n$ and all $z_{[0:n/2-1]}\in\Fq^{n/2}$.

By \eqref{eq:Lpm},
	\begin{equation*}
	\begin{aligned}
& W^+(-\alpha \cdot y_1,-\alpha^{-1}\cdot y_0,u_0 | u_1) = \frac 1 q  W(-\alpha \cdot y_1 |u_0 + \alpha u_1)W(-\alpha^{-1}\cdot y_0| u_1) \\
= & \frac 1 q  W( y_1 |-\alpha^{-1} u_0 - u_1) W(y_0|-\alpha u_1)
= \frac 1 q  W(y_0|-\alpha u_1) W( y_1 |-\alpha^{-1} u_0 - u_1) \\
= & W^+(y_0, y_1, u_0 | -\alpha^{-1} u_0-u_1)   
= W^+(-1\cdot (y_0, y_1, u_0) | \alpha^{-1} u_0 + u_1)  \\
= & W^+(-1\cdot (y_0, y_1, u_0)-\alpha^{-1} u_0| u_1)
		\end{aligned}
	\end{equation*}
	for all $y_0,y_1\in\cY$ and $u_0, u_1\in \Fq$.
In the Appendix, we prove that $W^+$ is also an $\Fq$-symmetric memoryless channel.
Therefore, for $a\in \Fq^*$ and $b\in\Fq$, we can use \eqref{eq:sth} to define the operation $a\cdot (y_0, y_1, u_0)+b$ on the output symbol $(y_0, y_1, u_0)$ of $W^+$. In particular, $-1\cdot (y_0, y_1, u_0)-\alpha^{-1} u_0$ in the above equations is defined in this way. Therefore, $T^+(-\alpha \cdot y_1,-\alpha^{-1}\cdot y_0,u_0) = -1\cdot T^+(y_0, y_1, u_0)-\alpha^{-1} u_0$.
Definition \eqref{eq:Lvecpm} further implies that
	\begin{equation} \label{eq:mpq}
		T_{[0:n/2-1]}^+( \xi_{m-1}^{(m)}(y_{[0:n-1]}),z_{[0:n/2-1]}) = -1\cdot T_{[0:n/2-1]}^+(y_{[0:n-1]},z_{[0:n/2-1]}) -\alpha^{-1}\cdot z_{[0:n/2-1]}
	\end{equation}
	for all $y_{[0:n-1]}\in\cY^n$ and all $z_{[0:n/2-1]}\in\Fq^{n/2}$.

	Let $|\cA|=k, |\cA^-|=k^-$, and  $|\cA^+|=k^+$.
	The condition \eqref{eq:condA} implies that $\cA^-\subseteq\cA^+$.
	As a consequence, $\Polar(n,k^-,\cA^-,0^{n/2-k^-}) \subseteq \Polar(n,k^+,\cA^+,0^{n/2-k^+})$.
	In other words, if $z_{[0:n/2-1]}$ is a codeword in $\Polar(n,k^-,\cA^-,0^{n/2-k^-})$, then it must also be a codeword in $\Polar(n,k^+,\cA^+,0^{n/2-k^+})$.

	We divide the proof into two cases.

		{\bf Case 1)} If $z_{[0:n/2-1]}$ is not a codeword in $\Polar(n,k^-,\cA^-,0^{n/2-k^-})$, then the probability on both sides of \eqref{eq:00} is $0$.
	This is simply because the SC decoder can only output a valid codeword as the decoding result.
	In this case, \eqref{eq:recx} implies that the probability on both sides of \eqref{eq:ktg} is $0$.

		{\bf Case 2)} If $z_{[0:n/2-1]}$ is a codeword in $\Polar(n,k^-,\cA^-,0^{n/2-k^-})$, then it is also a codeword in $\Polar(n,k^+,\cA^+,0^{n/2-k^+})$.
	In this case, we have
	\begin{align*}
		  & \bP\big( \hat{x}_{[0:n/2-1]}(T_{[0:n/2-1]}^+( \xi_{m-1}^{(m)}(y_{[0:n-1]}),z_{[0:n/2-1]} ),\cA^+) = -\alpha^{-1}\cdot(z_{[0:n/2-1]}+\alpha \cdot z_{[n/2:n-1]}) \big) \\
		= & \bP\big( \hat{x}_{[0:n/2-1]}(-1\cdot T_{[0:n/2-1]}^+(y_{[0:n-1]},z_{[0:n/2-1]}) -\alpha^{-1} z_{[0:n/2-1]},\cA^+) = - z_{[n/2:n-1]} -\alpha^{-1} z_{[0:n/2-1]}  \big)    \\
		= & \bP\big( \hat{x}_{[0:n/2-1]}(T_{[0:n/2-1]}^+(y_{[0:n-1]},z_{[0:n/2-1]} ),\cA^+) = z_{[n/2:n-1]} \big)
	\end{align*}
	for all $z_{[n/2:n-1]}\in\Fq^{n/2}$,
	where the first equality follows from \eqref{eq:mpq}, and the second equality follows from Lemma~\ref{lm:4}.
	Combining this with \eqref{eq:recx} and \eqref{eq:00}, we complete the proof of \eqref{eq:ktg}.
	Since \eqref{eq:ktg} is equivalent to \eqref{eq:r2}, this completes the proof of the lemma.
\end{proof}

\begin{lemma} \label{lm:6}
	Let $n=2^m$.
	Suppose that all the frozen symbols take value $0$.
	Suppose that the index set $\cA$ of information symbols satisfies the condition \eqref{eq:condA}.
	Then
	\begin{equation} \label{eq:r3}
		\bP\big( \hat{x}_{[0:n-1]}(y_{[0:n-1]},\cA) = x_{[0:n-1]}\big) = \bP\big( \hat{x}_{[0:n-1]}( \xi_{m-r}^{(m)}(y_{[0:n-1]}),\cA) =  \xi_{m-r}^{(m)}(x_{[0:n-1]}) \big)
	\end{equation}
	for all $y_{[0:n-1]}\in\cY^n$, all $x_{[0:n-1]}\in\Fq^n$, and all $1\le r\le m$.
\end{lemma}
\begin{proof}
	We prove by induction on $r$.
	The base case $r=1$ is already proved in Lemma~\ref{lm:5}.
	Now we assume that \eqref{eq:r3} holds for $r-1$ and all values of $m$, and we prove it for $r$.
	Observe that
	$$
		T_{[0:n/2-1]}^-( \xi_{m-r}^{(m)}(y_{[0:n-1]}))= \xi_{m-r}^{(m-1)}(T_{[0:n/2-1]}^-(y_{[0:n-1]}))  .
	$$
	Since $\cA$ satisfies the condition \eqref{eq:condA}, both $\cA^+$ and $\cA^-$ also satisfy the condition \eqref{eq:condA}.
	By the induction hypothesis, \eqref{eq:r3} holds for $r-1$ and $m-1$, so
	\begin{equation} \label{eq:1pt}
		\begin{aligned}
			  & \bP\big(\hat{x}_{[0:n/2-1]}(T_{[0:n/2-1]}^-(y_{[0:n-1]}),\cA^-) = z_{[0:n/2-1]}  \big)                                                                               \\
			= & \bP\big( \hat{x}_{[0:n/2-1]}( \xi_{m-r}^{(m-1)}(T_{[0:n/2-1]}^-(y_{[0:n-1]})),\cA^-) =  \xi_{m-r}^{(m-1)}(z_{[0:n/2-1]}) \big) \\
			= & \bP\big( \hat{x}_{[0:n/2-1]}(T_{[0:n/2-1]}^-( \xi_{m-r}^{(m)}(y_{[0:n-1]})),\cA^-) =  \xi_{m-r}^{(m-1)}(z_{[0:n/2-1]}) \big)
		\end{aligned}
	\end{equation}
	for all $z_{[0:n/2-1]}\in\Fq^{n/2}$.
	It is easy to verify that
	$$
		T_{[0:n/2-1]}^+( \xi_{m-r}^{(m)}(y_{[0:n-1]}),  \xi_{m-r}^{(m-1)}(z_{[0:n/2-1]}) )
		=  \xi_{m-r}^{(m-1)} ( T_{[0:n/2-1]}^+(y_{[0:n-1]}, z_{[0:n/2-1]} ) ) .
	$$
	Again by the induction hypothesis,
	\begin{align*}
		  & \bP\big( \hat{x}_{[0:n/2-1]}(T_{[0:n/2-1]}^+(y_{[0:n-1]}, z_{[0:n/2-1]} ),\cA^+) = z_{[n/2:n-1]} \big)                                                                                                                        \\
		= & \bP\big( \hat{x}_{[0:n/2-1]}(  \xi_{m-r}^{(m-1)} (T_{[0:n/2-1]}^+(y_{[0:n-1]}, z_{[0:n/2-1]} )),\cA^+) =  \xi_{m-r}^{(m-1)} (z_{[n/2:n-1]}) \big)                                     \\
		= & \bP\big( \hat{x}_{[0:n/2-1]}( T_{[0:n/2-1]}^+( \xi_{m-r}^{(m)}(y_{[0:n-1]}),  \xi_{m-r}^{(m-1)}(z_{[0:n/2-1]}) ),\cA^+) = \xi_{m-r}^{(m-1)} (z_{[n/2:n-1]}) \big)
	\end{align*}
	for all $z_{[n/2:n-1]}\in\Fq^{n/2}$.
	Combining this with \eqref{eq:recx} and \eqref{eq:1pt}, we obtain that
	\begin{align*}
		  & \bP\big(\hat{x}_{[0:n-1]}(y_{[0:n-1]},\cA) = (z_{[0:n/2-1]} +\alpha \cdot z_{[n/2:n-1]}, z_{[n/2:n-1]}) \big)                                                                                                                    \\
		= & \bP\big(\hat{x}_{[0:n-1]}( \xi_{m-r}^{(m)}(y_{[0:n-1]}),\cA) = ( \xi_{m-r}^{(m-1)}(z_{[0:n/2-1]} +\alpha \cdot z_{[n/2:n-1]}),  \xi_{m-r}^{(m-1)}(z_{[n/2:n-1]}) ) \big)
	\end{align*}
	for all $z_{[0:n-1]}\in\Fq^n$.
	Since
	\begin{align*}
		  &  (\xi_{m-r}^{(m-1)}(z_{[0:n/2-1]} +\alpha \cdot z_{[n/2:n-1]}) ,  \xi_{m-r}^{(m-1)}(z_{[n/2:n-1]}) ) \\
		= &  (\xi_{m-r}^{(m)} ( (z_{[0:n/2-1]} +\alpha\cdot  z_{[n/2:n-1]}, z_{[n/2:n-1]}) ),
	\end{align*}
	we further obtain that
	\begin{align*}
		  & \bP\big(\hat{x}_{[0:n-1]}(y_{[0:n-1]},\cA) = (z_{[0:n/2-1]} +\alpha \cdot  z_{[n/2:n-1]}, z_{[n/2:n-1]}) \big)                                                                    \\
		= & \bP\big(\hat{x}_{[0:n-1]}(  \xi_{m-r}^{(m)}(y_{[0:n-1]}),\cA) =   \xi_{m-r}^{(m)} ( (z_{[0:n/2-1]} +\alpha \cdot  z_{[n/2:n-1]}, z_{[n/2:n-1]}) ) \big)
	\end{align*}
	for all $z_{[0:n-1]}\in\Fq^n$.
	Finally, \eqref{eq:r3} follows from replacing $(z_{[0:n/2-1]} + \alpha \cdot z_{[n/2:n-1]}, z_{[n/2:n-1]})$ with $x_{[0:n-1]}$.
	This completes the proof of the lemma.
\end{proof}

\begin{figure*}
	\centering
	\begin{subfigure}{0.4\linewidth}
		\centering
		\includegraphics[width=\linewidth]{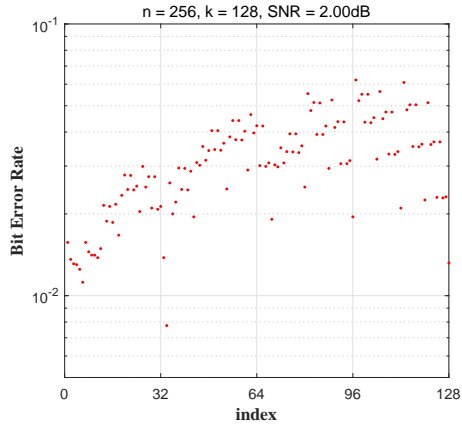}
		\caption{BER of $128$ message bits.
			Maximum value is $6.22\times 10^{-2}$.
			Minimum value is $0.78\times 10^{-2}$.}
	\end{subfigure}
	~\hfill
	\begin{subfigure}{0.4\linewidth}
		\centering
		\includegraphics[width=\linewidth]{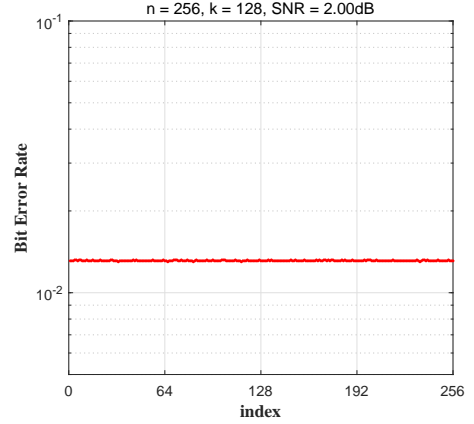}
		\caption{BER of $256$ codeword bits.
			Maximum value is $1.32\times 10^{-2}$.
			Minimum value is $1.30\times 10^{-2}$.}
	\end{subfigure}
	\caption{Simulation results of the $(n=256,k=128)$ binary polar code over binary-input AWGN channel with $E_b/N_0=2$dB}
	\label{fig:cp1}
\end{figure*}

\begin{lemma} \label{lm:7}
	Let $n=2^m$.
	Let $W$ be a $\ff_q$-symmetric memoryless channel.
	Suppose that the index set $\cA$ of information symbols satisfies the condition \eqref{eq:condA}.
	Then
	$$
		\SER_j(\Polar(n,k,\cA,\bar{u}_{\cA^c}),W) =
		\SER_{\delta_r^{(m)}(j)}(\Polar(n,k,\cA,\bar{u}_{\cA^c}),W)
	$$
	for all $0\le j\le n-1$, all $0\le r\le m-1$, and all $\bar{u}_{\cA^c}\in\Fq^{n-k}$.
\end{lemma}
\begin{proof}
	By \eqref{eq:uto0}, we have
	\begin{align*}
		  & \SER_j(\Polar(n,k,\cA,\bar{u}_{\cA^c}), W)                                                                                                                                      \\
		= & \sum_{y_{[0:n-1]}\in\cY^n} ~\sum_{x_{[0:n-1]}\in\Fq^n} W^n(y_{[0:n-1]}| 0^n) \mathbb{P}(\hat{x}_{[0:n-1]}(y_{[0:n-1]},\cA)=x_{[0:n-1]}) \mathbf{1}[x_j\neq 0]                   \\
		= & \sum_{y_{[0:n-1]}\in\cY^n} ~\sum_{x_{[0:n-1]}\in\Fq^n} \Big( W^n( \xi_r^{(m)}(y_{[0:n-1]}) | 0^n)                                                                 \\
		  & \hspace*{1.7in} \times \bP\big( \hat{x}_{[0:n-1]}( \xi_r^{(m)}(y_{[0:n-1]}),\cA) =  \xi_r^{(m)}(x_{[0:n-1]}) \big) \mathbf{1}[x_j\neq 0]  \Big)       \\
		= & \sum_{y_{[0:n-1]}\in\cY^n} ~\sum_{x_{[0:n-1]}\in\Fq^n} W^n(y_{[0:n-1]}| 0^n) \mathbb{P}(\hat{x}_{[0:n-1]}(y_{[0:n-1]},\cA)=x_{[0:n-1]}) \mathbf{1}[x_{\delta_r^{(m)}(j)}\neq 0] \\
		= & \SER_{\delta_r^{(m)}(j)}(\Polar(n,k,\cA,\bar{u}_{\cA^c}),W) ,
	\end{align*}
	where the second equality follows from $W^n(y_{[0:n-1]}| 0^n)= W^n(  \xi_r^{(m)}(y_{[0:n-1]}) | 0^n)$ and Lemma~\ref{lm:6}; the third equality is obtained from replacing $ \xi_r^{(m)}(y_{[0:n-1]})$ with $y_{[0:n-1]}$ and replacing $ \xi_r^{(m)}(x_{[0:n-1]})$ with $x_{[0:n-1]}$.
\end{proof}

Lemma~\ref{lm:7} immediately implies Theorem~\ref{thm:main} because for any $j\in\{0,1,\dots,n-1\}$, we can always apply a subset of $\delta_0^{(m)},\delta_1^{(m)},\dots,\delta_{m-1}^{(m)}$ to $j$ and obtain $0$.
This means that $\SER_j(\Polar(n,k,\cA,\bar{u}_{\cA^c}),W) =
	\SER_0(\Polar(n,k,\cA,\bar{u}_{\cA^c}),W)$ for all $j\in\{0,1,\dots,n-1\}$ and completes the proof of Theorem~\ref{thm:main}.


\section{The connection to \cite{Arikan11}} \label{sect:4}

In this section, we restrict ourselves to binary polar codes.
As mentioned in Introduction, we observed the simulation results in Fig.~\ref{fig:cp1} as we repeated the experiments in \cite{Arikan11}.
From Fig.~\ref{fig:cp1} we can see that the variance of the BERs of message bits is very large.
In contrast, the BERs of codeword bits are extremely stable, to the extent that they form a straight line.
At this point, it makes sense to discuss the connection between our results and \cite{Arikan11}.

We still use $u_{[0:n-1]}$ to denote the message vector and use $x_{[0:n-1]}$ to denote the codeword vector.
The set $\cA$ is still the index set of information bits.
The main observation in \cite{Arikan11} is that the average BER of $x_{\cA}$ is much smaller than the average BER of $u_{\cA}$ under the SC decoder.
This is somewhat counter-intuitive because the SC decoder directly decodes $u_{\cA}$, and $x_{\cA}$ is obtained from multiplying the decoding result of $u_{[0:n-1]}$ with the encoding matrix $\mathbf{G}_n$.
Even till today, no rigorous analysis is available to explain this phenomenon.

The results in this paper imply that the BERs of all the codeword bits are equal to each other, and they all equal to the BER of the last message bit because the last codeword bit is the same as the last message bit.
On the bright side, the last message bit is the best-protected message bit if we assume that all the previous message bits are decoded correctly.
However, we also have an argument in the opposite direction: The decoding error in the SC decoder accumulates, and the last message bit takes the most damage from decoding errors in previous message bits.
In fact, if we look at Fig.~\ref{fig:cp1} closely, we can see that there is an increasing trend in the BERs of message bits (increasing with the indices).
The BER of the last message bit is an outlier because it drops abruptly compared to the previous bits.

To conclude, although our result is somewhat correlated with the observation in \cite{Arikan11}, we are still not able to explain the phenomenon in \cite{Arikan11}.
That calls for future research effort.

\section*{Acknowledgement}
We thank Ling Liu and Henry Pfister for providing valuable feedback on an earlier version of this paper.

\bibliographystyle{IEEEtran}
\bibliography{BER}

\appendix

\begin{lemma}
If $W$ is an $\Fq$-symmetric memoryless channel, then both $W^+$ and $W^-$ defined in \eqref{eq:Lpm} are $\Fq$-symmetric memoryless channels.
\end{lemma}

\begin{proof}
We first prove the claim for $W^+$. By definition, we need to show that (1) there exist $q$ permutations $\{\sigma_b: b\in\Fq\}$ on the output alphabet $\cY^2\times \Fq$ such that $W^+(y_0,y_1,u_0|u) = W^+(\sigma_{u'-u}((y_0,y_1,u_0))|u')$ for all $y_0,y_1\in\cY$ and all $u,u',u_0\in\Fq$;
(2) there exist $q-1$ permutations $\{\pi_a: a\in\Fq^*\}$ on $\cY^2\times \Fq$ such that $W^+(y_0,y_1,u_0|u) = W^+(\pi_a((y_0,y_1,u_0))|a u)$ for all $y_0,y_1\in \cY$, $u,u_0\in\Fq$, and $a\in\Fq^*$.

We may choose $\sigma_b((y_0,y_1,u_0))=(y_0+\alpha b, y_1+b, u_0)$ and $\pi_a((y_0,y_1,u_0))=(a\cdot y_0, a\cdot y_1, a u_0)$. The following calculations allow us to verify this choice.
\begin{align*}
& W^+(\sigma_b((y_0,y_1,u_0))|u+b)
=W^+(y_0+\alpha b, y_1+b, u_0|u+b) \\
= & \frac 1 q W(y_0+\alpha b|u_0 + \alpha (u+b))W(y_1+b| u+b)
= \frac 1 q W(y_0|u_0 + \alpha u)W(y_1| u) \\
= & W^+(y_0,y_1,u_0 | u) , \\
& W^+(\pi_a((y_0,y_1,u_0))|au) = W^+(a\cdot y_0, a\cdot y_1, a u_0| au) \\
= & \frac 1 q W(a\cdot y_0|au_0 + \alpha au)W(a\cdot y_1| au)
= \frac 1 q W(y_0|u_0 + \alpha u)W(y_1| u) \\
= & W^+(y_0,y_1,u_0 | u) .
\end{align*}
As for $W^-$, we need to show that (1) there exist $q$ permutations $\{\sigma_b: b\in\Fq\}$ on the output alphabet $\cY^2$ such that $W^-(y_0,y_1|u) = W^-(\sigma_{u'-u}((y_0,y_1))|u')$ for all $y_0,y_1\in\cY$ and all $u,u'\in\Fq$;
(2) there exist $q-1$ permutations $\{\pi_a: a\in\Fq^*\}$ on $\cY^2$ such that $W^-(y_0,y_1|u) = W^-(\pi_a((y_0,y_1))|a u)$ for all $y_0,y_1\in \cY$, $u\in\Fq$, and $a\in\Fq^*$.

We may choose $\sigma_b((y_0,y_1))=(y_0+ b, y_1)$ and $\pi_a((y_0,y_1))=(a\cdot y_0, a\cdot y_1)$. The following calculations allow us to verify this choice.
\begin{align*}
& W^-(\sigma_b((y_0,y_1))| u+b) = W^-(y_0+ b, y_1| u+b) \\
= & \frac 1 q \sum_{u_1\in\Fq} W(y_0 +b|u+b + \alpha u_1)W(y_1| u_1) \\
= & \frac 1 q \sum_{u_1\in\Fq} W(y_0|u + \alpha u_1)W(y_1| u_1) 
=  W^-(y_0,y_1  | u) , \\
& W^-(\pi_a((y_0,y_1)) | au) = W^-(a\cdot y_0, a\cdot y_1| au) \\
= & \frac 1 q \sum_{u_1\in\Fq} W(a\cdot y_0| au + \alpha u_1)W(a\cdot y_1| u_1) \\
= & \frac 1 q \sum_{u_1\in\Fq} W(a\cdot y_0| au + \alpha au_1)W(a\cdot y_1| au_1) \\
= & \frac 1 q \sum_{u_1\in\Fq} W(y_0|u + \alpha u_1)W(y_1| u_1) 
=  W^-(y_0,y_1  | u) .
\end{align*}
\end{proof}

\end{document}